\documentclass[final,1p,times]{elsarticle}

\usepackage{tikz}

\usepackage[hyperindex,breaklinks]{hyperref}
\usepackage{breakurl}
\hypersetup{
colorlinks= false,
linkcolor=black,
citecolor=black,
citecolor= gray
}\usepackage[normalem]{ulem}

\usepackage{longtable}
\usepackage{stfloats}
\usepackage{subcaption}
\usepackage{latexsym}
\usepackage{amsmath}
\usepackage{bm}
\usepackage{amsthm} 
\usepackage{bbm}
\usepackage{afterpage}
\usepackage{multirow}
\usepackage{tabularx}
\usepackage{algorithm}
\usepackage{xcolor,colortbl}
\usepackage{algpseudocode}
\usepackage{booktabs}
\usepackage{dsfont}
\usepackage{pdflscape}
\usepackage{afterpage}
\usepackage{comment}
\usepackage{enumitem}
\usepackage{url}
\graphicspath{{figure/}}
\journal{arXiv}
\biboptions{authoryear}

\DeclareMathOperator*{\argmin}{argmin}
\begin{document}

\newtheorem{proposition}{Proposition}
\newtheorem{lemma}{Lemma}
\newdefinition{remark}{Remark}
\newproof{pf}{Proof}
\newtheorem{theorem}{Theorem}
\def\UrlBreaks{\do\/\do-}
\newcommand{\tabincell}[2]{\begin{tabular}{@{}#1@{}}#2\end{tabular}}
\allowdisplaybreaks
\begin{frontmatter}
 \title{Dynamic Vaccination Game in A Heterogeneous Mixing Population}


\author[CEE]{Liqun Lu}
\author[CEE]{Yanfeng Ouyang\corref{correspondingauthor}}
\address[CEE]{Department of Civil and Environmental Engineering, University of Illinois at Urbana-Champaign, Urbana, IL 61801, USA}
\cortext[correspondingauthor]{Corresponding author. Tel: +1 217 333 9858}
\ead{yfouyang@illinois.edu}

\begin{abstract}
	Opposition to vaccination has long been a non-negligible public health phenomenon resulted from people's varied perceptions toward vaccination (e.g., vaccine-phobia). This paper investigates the voluntary vaccination behavior of a heterogeneous population during an epidemic outbreak, where each individual makes its own vaccination decision to minimize its expected disutility from both vaccine-phobia and the risk of infection. Such a problem is known as a vaccination game, as people's vaccination decisions not only affect their own disutilities but those of all others through probabilistic disease transmissions. To study the vaccination game, the susceptible-infected-removed disease propagation process is generalized into a new epidemic dynamics model to allow dynamic vaccination and immunity activation in a heterogeneous mixing population. An efficient computation method is proposed to evaluate the final state of the dynamic epidemic system. Then, a classic game-theoretical equilibrium model is built upon these results to examine the impacts of people's vaccination behavior on the overall risk of epidemic outbreak. A hypothetical case study is used to validate the dynamics model and the derived results, and extensive numerical experiments are conducted to identify the key factors that affect people's vaccination decisions and the risk of an outbreak. Moreover, three alternative vaccination schemes are also studied to examine the effects of early and non-differential vaccination treatments, respectively.
\end{abstract}


\end{frontmatter}
\section{Introduction}\label{sec:Introduction}
Since the invention of the first vaccine against smallpox in 1798, a great many different vaccines have been developed as the most powerful weapons against infectious diseases. Besides the smallpox, a handful of the other diseases that had haunted human beings for a long history, including diphtheria, measles, hepatitis, etc., are now finally under control. Nonetheless, while the human society are freed from the fear of infectious diseases in the modern world, people start to cast vaccines aside as an option rather than a necessity \citep{history}. It has been reported that some kids in the U.S. are again exposed to high-risk infectious diseases such as measles and meningitis because their parents opted out of public vaccination programs \citep{sun_2018}; moreover, an online poll showed that over one third of the U.S. parents do not intend to give flu-shots to their school-age children in the 2018 flu season \citep{flushots}. 
One of the reported reasons of these opt-outs is the disputable link between the vaccination and child autism \citep{calandrillo2003vanishing}. Furthermore, numerous surveys have shown that during the recent H1N1 pandemic, the public's intention of vaccine uptake is disappointingly low. This is even the case for healthcare workers themselves who have above-average knowledge of vaccinations \citep{myers2011determinants,chor2009willingness, tozzi2009parental}. These reports show that ``vaccine-phobia'' is not merely a subculture among a certain group of people, but rather a serious issue that demands attention in the public health industry.
\cite{mckee2016exploring} summarized four main reasons for people of avoid vaccination: religious concerns, philosophical concerns, safety concerns, and insufficient knowledge. As such, the idea that one may be vaccinated at some point (even before the actual event of vaccination) often causes psychological sufferings.\footnote{Most existing work that investigated people's vaccination behavior often assumed people incur penalty only after they have been vaccinated, which is suitable when vaccines have non-negligible side-effects or may lead to morbidity. Yet the psychological sufferings could occur regardless of the actual presence of negative consequences.} Yet the potentially disastrous outcome of infection drives a person toward accepting vaccination in the meantime. Therefore, for each individual the decisions on whether and when to take vaccine depend on the trade-off between the psychological penalty and the risk of being infected. Moreover, in a population, the infection risk also depends on the overall vaccination coverage of all individuals -- if all others with whom one interacts are all vaccinated, then there is no fear about the epidemic. Therefore, the decisions of all individuals are simultaneously influenced by, and also influencing, each other. This problem is referred to as a \textit{vaccination game} in the literature. It is important to study the complex behavior of individuals in a vaccination game so as to understand the development and consequence of an epidemic outbreak, and in turn to reveal insights on epidemic controls.

Modeling the vaccination game is challenging, because it is imperative to account for heterogeneity among the individuals. While the disease characteristics and the exogenous environment, e.g., the disease propagation and the vaccination effectiveness conditions, can be relatively similar to everyone, each individual has a different physical condition and psychological belief system. The interactions between individuals themselves and those with the environment also vary drastically from person to person; e.g., depending on how frequently a person makes disease-transmissive contacts with the others. Consequently, different people often make different vaccination decisions. Furthermore, it shall be noted that a person with more interactions not only has a higher infection risk himself/herself, but also poses a greater threat all those around him/her. In other words, people's vaccination decisions contribute differently to the epidemic outbreak, and such heterogeneity should be taken into consideration.

Another challenging issue associated with a vaccination game is the proper modeling of the epidemic evolution process. Since disease propagation and vaccination decisions are dynamic, epidemic dynamic models must be developed to evaluate the final or steady state of the system. 
The well-known susceptible-infected-removed (SIR) models usually demand tedious iterations over time steps to evaluate the outcome of a vaccination strategy; e.g., see \cite{bauch2003group} for an example. In the hope to avoid such difficulty, most of the existing literature on dynamic vaccination game either (i) considered susceptible-infected-susceptible problems \citep{reluga2009sis} or a dynamic population \citep{bauch2004vaccination, reluga2006evolving} such that a steady state equilibrium of the epidemic can be found, or (ii) approximated the results by using a perceived infection risk among the population \citep{cojocaru2008dynamic} or linearized system dynamics \citep{reluga2011general}. A method to conveniently and accurately evaluate the outcome of a vaccination strategy during a fast-spreading SIR epidemic is yet to be found.

In light of these challenges, we investigate the dynamic vaccination game in a large heterogeneous population. We consider the infectious disease spreads quickly via person-to-person contacts in a population contact network. To stay focused, we consider population heterogeneity in terms of a person's number of contacts, while every individual shares the same level of vaccine-phobia. The main contributions of this paper are summarized as follows:

\begin{enumerate}[noitemsep,nolistsep]
	\item A susceptible-infected-removed-vaccinated-activated (SIRVA) epidemic dynamic model is developed to describe the coupled disease propagation process and the vaccination dynamics; a computation method to conveniently evaluate the final outcome of people's vaccination decisions is proposed;
	\item Then, the vaccination game in a large heterogeneous population is formulated into a mathematical program, to which existence and uniqueness of an equilibrium solution are investigated; based on these findings, an efficient heuristic algorithm is proposed to solve the equilibrium solution to the game;
	\item A hypothetical case study is performed to validate the dynamic model and the derived results, and to demonstrate the vaccination game model and solution approach; through extensive numerical experiments, we investigate several critical factors (including network connectivity, disease transmissibility, and vaccine-phobia level) that affect people's vaccination behavior and the epidemic propagation outcome; moreover, the impacts of delayed and homogeneous vaccination are examined as well.
\end{enumerate}

The rest of this paper is organized as follows. The next section reviews some of the most relevant literature on epidemic dynamic models and vaccination games. Section \ref{sec:methodology} presents the developed SIRVA dynamic model and the analytical results on the final epidemic size; some properties of the vaccination game equilibrium are also presented, based on which a heuristic solution approach is proposed to find the equilibrium. Next, Section \ref{sec: other schemes} discusses several different vaccination schemes, including delayed and non-differential vaccinations. Section \ref{sec:numerical} presents the case study, the sensitivity analysis, and discusses various insights. Finally, Section \ref{sec:Conclusion} concludes this paper and discusses future research directions. 

\section{Literature Review} \label{sec:Literature}
\subsection{Epidemic dynamic modeling}\label{sec:Literature-dynamic modeling}
Infectious disease modeling has been intensively studied in the past century since the establishment of the fundamental Kermack-Mckendrick model \citep{kermack1927contribution}; see \citep{anderson1992infectious} for a review. It categorizes the population into several compartments based on their health status: susceptible (S), infectious/infected (I) and recovered/removed (R), and hence is also called the SIR model. The Kermack-Mckendrick model, and many others alike, assume homogeneous mixing of the population; however, many infectious diseases are transmitted via person-to-person contact networks, which can be highly heterogeneous.

Numerous network-based epidemic models have been developed with a variety of focuses; see \cite{pastor2015epidemic} for a summary. In particular, random network models have gained much attention since the emergence of random network analysis on configuration models \citep{newman2001random}. In these models, each individual in the population is represented by a node, and an (undirected) edge represents a direct contact between two individuals that may transmit the disease. The actual connections among the nodes are described probabilistically; i.e., each node is associated with a degree of \textit{stubs}, and an edge is formed by randomly connecting two unconnected stubs. The degree of an arbitrary node follows a given probability distribution, and such \textit{degree distribution} is expressed in the form of a probability generating function, as follows \citep{molloy1995critical}:
\begin{equation}
g(x) = \sum_{k} p_k x^k,
\end{equation}
where $p_k$ is the probability that a random node has $k$ degrees, and $x$ is a dummy argument. As such, $g'(1)$ and $g''(1)$  are the expected numbers of direct neighbors and second-order neighbors of an arbitrary node, respectively.

\cite{volz2008sir} and \cite{miller2011note} proposed a system of low-dimensional equations to describe the SIR epidemic dynamics in a population with heterogeneous mixing, modeled as a configuration model random network. This model tries to capture the probability for a randomly selected test node to be susceptible, based on whether or not any of its incident edges have transmitted disease from its (infected) neighbors. For convenience, we call such an event as an \textit{infectious contact}. Two variables are defined to capture this quantity:
\begin{itemize}[noitemsep,nolistsep]
	\item $\theta$: the probability that a random edge of the test node has not passed an infectious contact to the test node, assuming the test node does not transmit disease to others;
	\item $\phi$: the probability that a random edge of the test node is connected to an infected neighbor, but has not passed an infectious contact to the test node, assuming the test node does not transmit disease to others.
\end{itemize}
An infectious contact along any edge and the recovery of any infected node are both assumed to follow Poisson processes, at a unit transmission rate, denoted $r$, and a recovery rate, denoted $u$, respectively. Moreover, $S$, $I$, and $R$ are used to represent the fractions of population in the three compartments, respectively. The system equations can be summarized as follows:
\begin{align}
& \dot{\theta} = - r \phi, \forall t \geq 0,  \label{eq: dot theta natural} \\
& \dot{\phi} = \left( -r - u + r \frac{g''(\theta) }{g'(1)} \right)  \phi, \forall t \geq 0,  \label{eq: dot phi natural} \\
& \dot{R} = u I, \forall t \geq 0,  \label{eq: dot R natural} \\
& S(t)  = g \left( \theta \right), \forall t \geq 0, \label{eq: S natural} \\
& S(t) + I(t) + R(t) = 1, \forall t \geq 0. \label{eq: SIR} 
\end{align}
In Eq.\eqref{eq: dot phi natural}, the first two terms in the parentheses on the right hand side represent, respectively, the transmission of the disease along an edge that satisfies the definition of $\phi$, and the removal of an infected neighbor linked by such an edge. The third term is the derivative of $ -g'(\theta) / g'(1)$, where $g'(\theta) / g'(1)$ represents the probability that an arbitrary neighbor of the test node is susceptible. Therefore, the third term captures the increment of $\phi$ due to new infections at the neighbors of the test node (i.e., infected by their own neighbors other than the test node). This model is simple but effective, and hence will be used as the basis for the proposed model in this paper. 

\subsection{Vaccination control}\label{sec:Literature-vaccination control}
The effect of vaccination could be modeled in a variety of ways, mainly depending on its timing: i) vaccination prior to the epidemic outbreak, or immediate vaccination once a new population is born (e.g., \cite{madar2004immunization} and \cite{bauch2004vaccination}); and ii) vaccination during disease propagation when the vaccine takers have been exposed to potential risks of infection for some time (e.g., \cite{lu2002effect} and \cite{pang2007delayed}). It is worth noting that, if immunity obtained via vaccination does not wane, early vaccination is always more effective in mitigating disease propagation because of survival bias.

\cite{bansal2006comparative} and \cite{mylius2008optimal} compared two vaccination strategies under different vaccination timings, targeting subpopulations that have high infection risks (morbidity-based) and those more likely to decease once infected (mortality-based), respectively. Their results showed that if vaccines are available prior to the onset of disease outbreak, morbidity-based vaccination should be adopted as it immunizes the high risk population and mitigates disease propagation most effectively. However, delay in availability of vaccination has a tremendous impact on the epidemic result, and in that case mortality-based vaccination out-performs its morbidity-based counterpart. 

Impacts of network topology and vaccination strategy on SIR-type disease propagation in large-scale random networks were studied in \cite{ma2013importance}. Stochastic simulations were conducted for various combinations of random networks (e.g., Erdos-Renyi model, scale-free, etc.) and vaccination strategies (e.g., prioritized, random, etc.). They discovered that network topology has a significant impact on the spread of disease as well as the timing of vaccination. This finding highlights the importance of understanding disease transmission in realistic contact networks.
The recent independent effort in \cite{di2018multiple} incorporated dynamic vaccination into the models of \cite{volz2008sir} and \cite{miller2011note}. Despite some similarity in the research goals, there are key differences in our approaches. For example, their model did not account for population's heterogeneous vaccination behavior, which nonetheless is crucial in modeling the vaccination game in a heterogeneous mixing population. 
\subsection{Vaccination game}\label{sec:Literature-vaccination game}
The term ``vaccination game'' has been used to describe an imitation process where individuals learn about others' perspective on vaccine and follow probabilistically each other's actions \citep{bauch2005imitation, zhang2013impact, han2016evolutionary}. In this paper, we go back to its literal meaning of a classic game-theoretical framework, where each individual in a population is a decision maker that tries to maximize its own utility based on its knowledge on the system. A review on the impacts of typical vaccination behavior can be found in \cite{funk2010modelling}. However, the sources of such behavior as a result of people's decision-making processes in an epidemic outbreak are still not fully understood.

\cite{reluga2011general} proposed a general framework for vaccination game in a large population, where people's behavior was modeled as a Markovian decision process that maximizes the expected utility of each individual. Within this framework, they compared analytical solutions to the game under several scenarios, such as differential waning of immunity and imperfect immunity. Their model assumed that all susceptible individuals will eventually be vaccinated (or infected and removed if infection happens before vaccination), which is, however, not practical because people's vaccination intention during an outbreak would likely depend on the disease propagation status. 
\cite{bauch2004vaccination} investigated a problem where parents need to determine whether or not to vaccinate their children. By considering the equilibrium state in a dynamic population, they were able to directly compute the infection risk given a vaccination strategy. Moreover, through analytical discussions, they showed that the vaccination decisions in a homogeneous population shall be equal. Population heterogeneity was considered in \cite{cojocaru2008dynamic}, where multiple groups of population each has a different perception on the infection risk. Nash equilibria was formulated into evolutionary variational inequalities. Although the ground-truth infection risk in this work was simplified by a perceived risk, their results highlighted the importance of accounting for social heterogeneity in a vaccination game. 

Moreover, to investigate the vaccination game, it is imperative to quantitatively measure people's disutility due to psychological suffering, based on not only the likelihood of vaccine uptake, but also how people perceive vaccination (e.g., the level of vaccine-phobia). Existing social studies try to find the correlation between people's perception on vaccines and their vaccine uptake decisions from field surveys, e.g. \cite{smith2011parental}. In light of this, we will next aim at drawing a bridge between these two aspects of vaccination, by establishing a model on how vaccine-phobia may affect people's decision-making process and in turn the disease propagation. 

\section{Methodology}\label{sec:methodology}
\subsection{Vaccination Dynamics}\label{subsec: vaccination dynamics}
We consider a dynamic vaccination game in a large population with heterogeneous degree distributions, where the natural disease propagation follows an SIR process. In order to accommodate the vaccination process, we further define two compartments based on individuals' health status: vaccinated (V) and activated (A). A susceptible node becomes vaccinated once being injected with a vaccine and thus obtaining immunity. When a vaccinated node receives an infectious contact, instead of becoming infected (like a susceptible node), it activates its immunity and thus become activated. The transition among the different health states are summarized below:
\begin{itemize}[noitemsep,nolistsep]
	\item A susceptible node becomes infected upon receiving an infectious contact from one of its neighbors;
	\item An infected individual deceases or recovers by itself (with immunity) following the natural removal rate;
	\item A susceptible node becomes vaccinated at a certain rate (based on its own decision process);
	\item A vaccinated node becomes activated upon receiving an infectious contact from one of its neighbors.
\end{itemize}

We now use $S(t), I(t), R(t), V(t), A(t)$ to denote the fractions of susceptible, infected, removed, vaccinated, and activated individuals, respectively.\footnote{It might appear cumbersome to differentiate the activated from the vaccinated as they both represent the same immunized status. However, it is quite necessary to do so, because the model described by \eqref{eq: SIR} and \eqref{eq: dot theta natural} - \eqref{eq: S natural} focuses on the status of edges (i.e., whether or not they have passed infectious contacts), and hence the status of the nodes can only be derived from the edge status.}
The size of population is fixed, such that the equivalent of Eq.\eqref{eq: SIR} now becomes:
\begin{equation}
S(t) + I(t) +R(t) + V(t) + A(t)= 1, \forall t \geq 0.
\label{eq: fixed popualtion size}
\end{equation}
Furthermore, we differentiate nodes by their degrees; denote $X_k$ as the fraction of degree-$k$ nodes in compartment $X \in \left\lbrace S, I, R, V, A \right\rbrace$. Naturally, we have
\begin{subequations} \label{eq: degree does not change} 
	\begin{eqnarray}
	& \sum_{k=1}^K X_k (t)= X(t), \forall t \geq 0,  X \in \left\lbrace S, I, R, V, A \right\rbrace, \\
	& \sum_{X \in \left\lbrace S, I, R, V, A \right\rbrace} X_k(t) = p_k, \forall t \geq 0, k = 1, \dots ,K.
	\end{eqnarray}
\end{subequations}
where $K \ge 0$ is the maximum possible degree of a node. 

In the vaccination game, we consider the Nash game; i.e., each individual tries to minimize its own expected penalty by deciding its vaccination rate, given others' vaccination decisions. We use $v_k$ to denote the average vaccination rate of degree-$k$ nodes. The expected penalty includes the psychological penalty from vaccine-phobia (if vaccinated), and the penalty of eventually getting infected (if not vaccinated). 

Consider the status of a degree-$k$ node who has not received any infectious contacts from its neighbors. The fraction of such nodes is given by $p_k \theta^k$, and these nodes can either be susceptible or vaccinated, which yields the following equation:
\begin{equation}
S_k(t) + V_k (t) = p_k \theta^k, \forall t \geq 0,  k = 1, \dots, K. \label{eq: degree k non-tranmitted}
\end{equation}
The dynamics of $S_k$ contain two parts: the vaccinations and the infections of the degree-$k$ population, respectively. The former is simply $- v_k S_k$. Moreover, the test node being vaccinated or not does not affect the infectious contacts it receives, and thus the node being vaccinated or susceptible is random (as long as there is no infectious contact yet). Therefore, the latter is the rate of reduction among those who have not received infectious contacts, $d \left( p_k \theta^k \right) / dt $, times the fraction of susceptible nodes, $S_k/p_k \theta^k $. Furthermore, the dynamics of $\theta$ still follow Eq.\eqref{eq: dot theta natural} by definition. As such, we have
\begin{equation}
\dot{S}_k =  - v_k S_k + \frac{d \left( p_k \theta^k \right) }{ dt} \frac{S_k}{p_k \theta^k} =\left( - v_k - k r \frac{\phi}{\theta} \right)  S_k, \forall t \geq 0, k = 1, \dots, K. \label{eq: dot S}
\end{equation}

The dynamics of $\phi$, i.e. Eq.\eqref{eq: dot phi natural}, should be slightly modified. Recall its dynamics include three parts: (i) transmission of disease along edges to the test node, (ii) removal of the infected neighbors, and iii) new infections of the susceptible neighbors. Vaccination of the susceptible does not affect the already infected nodes nor occurrence of infectious contacts, thus the first two terms remain unchanged. The third term is essentially the rate at which a susceptible neighbor of the test node becoming infected, i.e., infectious contacts happening to the neighbor while the neighbor is not vaccinated. 
The probability that a random neighbor of the test node is a degree-$k$ node who has not received any infectious contacts from others is $ k p_k \theta^{k-1} / g'(1)$. The rate at which infectious contacts happen to this neighbor is then $d\left(- k p_k \theta^{k-1} / g'(1) \right) / dt$. Using the same argument that leads to Eq. \eqref{eq: dot S}, the rate of the neighbor being infected is $-\left(  S_k / p_k \theta^k \right) \cdot d\left(k p_k \theta^{k-1} / g'(1) \right) / dt$. Therefore, Eq.\eqref{eq: dot phi natural} becomes
\begin{equation}
\dot{\phi} = \left( -r - u + \frac{r}{g'(1) \theta^2 } \sum_{k=1}^K S_k k (k-1) \right) \phi . \forall t \geq 0.  \label{eq: dot phi with vac 1}
\end{equation}

Finally, a degree-$k$ vaccinated node becomes activated at the same rate as a degree-$k$ susceptible node becomes infected. Therefore,
\begin{equation}
\dot{A}_k = k r \frac{\phi}{\theta} V_k, \forall t \geq 0,  k = 1, \dots, K, \label{eq: A} 
\end{equation}

To this point, we have formulated the complete system dynamics with Eq.\eqref{eq: dot theta natural}, \eqref{eq: dot R natural}, and \eqref{eq: fixed popualtion size}-\eqref{eq: A}. One can evaluate the impacts of the population's vaccination decisions by iteratively computing the dynamic equations over time from the beginning of the epidemic to the end. This is computationally burdensome for a large population and cannot easily reveal insights on the vaccination game. In light of this, we further assume that people's vaccination strategies $v_k$ follow a special form, such that the outcome of the epidemics can be derived analytically in closed form. This will be the focus of the remainder of this paper.

\subsection{Final Epidemic Size} \label{subsec: final epidemic size}
In reality, people often determine their vaccination rate based on the current propagation rate of the disease. Therefore, we consider a vaccination strategy in which the vaccination rate is proportional to the fraction of infectious contacts that is going to happen. From now on, we use an explicit argument $t$ for all the time-dependent variables so as to distinguish them from time-invariant parameters. Then the vaccination strategy is presented as follows:
\begin{equation}
v_k(t) = \mu_k r \frac{\phi (t)}{\theta(t)}, \forall t \geq 0, k = 1, \dots, K, \label{eq: special v}
\end{equation}
where $\mu_k \geq 0$ is a constant parameter determined by the degree-$k$ population. It reflects the likelihood of the population taking vaccines, and we call it \textit{vaccine adoption level}.\footnote{Here $\mu_k$ represents an averaged vaccine adoption level to capture the overall vaccination rate of the degree-$k$ population. Yet in fact each individual could make its own decision and their vaccine adoption levels may differ. However, as we will show later, individual decisions should be equal under equilibrium for nodes with the same degree. Thus it is safe to use an aggregated rate in the following derivations.} The higher vaccine adoption level $\mu_k$ among the population, the lower infection risk but a greater psychological suffering (due to the vaccine-phobia).

With Eq. \eqref{eq: special v}, we have the following equation:
\begin{equation}
S_k(t) = p_k \theta(t)^{k + \mu_k}, \forall t \geq 0. \label{eq: special S}
\end{equation}
Readers can easily verify this by taking derivative over Eq. \eqref{eq: special S} and comparing it with Eq. \eqref{eq: dot S}. 

Then we use the same approach as in \cite{miller2011note} to obtain the final epidemic size. First, we denote $\theta(\infty)$ as $\theta_\infty$, and the total immunized fraction as $M$, i.e., $M = V(\infty) + A(\infty) = \sum_k \int_{0}^{\infty} v_k(\tau) S_k (\tau) d\tau $. With Eq. \eqref{eq: special v} and \eqref{eq: special S}, we have
\begin{equation}
M = \sum_{k=1}^K \frac{\mu_k p_k}{k + \mu_k} \left( 1 - \theta_{\infty}^{k+\mu_k} \right). \label{eq: p special}
\end{equation}
By the end of the epidemic, the degree-$k$ uninfected population is the summation of susceptible, vaccinated, and activated nodes:
\begin{equation}
S_k(\infty) + V_k(\infty) + A_k(\infty) = \frac{ p_k k \theta_{\infty}^{k+\mu_k} + \mu_k p_k }{k+\mu_k}, \forall k = 1, \dots, K  . \label{eq: total healthy}
\end{equation}

Moreover, we let $\pi$ be the probability that a random neighbor of the test node has never been infected by the end of the epidemic, given that the test node does not transmit the disease. This could happen under two mutually exclusive events: (i) no infectious contact has ever happened to this neighbor, and (ii) this neighbor has already been vaccinated when an infectious contact arrives, i.e. this neighbor is activated. The first probability is given by $g'(\theta_\infty) / g'(1)$. The second can be computed via integrating $ \sum_{k=2}^{K} \frac{V_k}{p_k \theta^k} d \left( - \frac{ p_k k \theta^{k-1} }{g'(1)} \right)/ dt $ from $t=0$ to $\infty$. The result is presented as follows:
\begin{equation}
\pi = \frac{1}{g'(1)} \left( g'(\theta_\infty) + \sum_{k=2}^K p_k k \left( \frac{\mu_k}{k + \mu_k -1} - \theta_\infty^{k-1} + \frac{k -1}{k + \mu_k - 1} \theta_\infty^{k + \mu_k -1} \right)  \right). \label{eq: P}
\end{equation}

We further denote the \textit{transmissibility} of disease as $T = \frac{r}{r+u}$. It gives the probability that, conditional on the fact that one end of a random edge is infected, an infectious contact happens along this edge before that infected end is removed (recall, that both events follow Poisson processes). The probability that a random edge of the test node has never transmitted an infectious contact is then $1 - T ( 1- \pi)$, which, from Eq. \eqref{eq: P}, yields the following nonlinear equation:
\begin{equation}
\theta_\infty = 1 - T + \frac{T}{g'(1)} \left( g'(\theta_\infty) +  \sum_{k=2}^K p_k k \left( \frac{\mu_k}{k + \mu_k -1} - \theta_\infty^{k-1} + \frac{k -1}{k + \mu_k - 1} \theta_\infty^{k + \mu_k -1} \right) \right). \label{eq: theta_infty}
\end{equation}
The solution can be easily found by searching from 1 to 0 with a small step size. Then, the final uninfected population, from Eq. \eqref{eq: total healthy}, as well as the final epidemic size, $R(\infty) = 1 - (S(\infty) + V(\infty) + A(\infty))$, can be computed accordingly. Moreover, the fraction of nodes in other compartments in the final state can be computed as follows:
\begin{align}
& A_k(\infty) = p_k \left( \frac{\mu_k}{k+\mu_k}  - \theta_\infty^k +  \frac{k}{k+\mu_k} \theta_\infty^{k+\mu_k} \right), \forall k = 1, \dots, K,\label{eq: special final Ak}\\
& I_k(\infty) = 0, \forall k = 1, \dots, K, \label{eq: special final Ik} \\
& R_k(\infty) = p_k \left( 1 - \frac{ k \theta_{\infty}^{k+\mu_k} + \mu_k }{k+\mu_k} \right) , \forall k = 1, \dots, K.  \label{eq: special final Rk} 
\end{align}

\subsection{Vaccination Game}\label{sec:Game}
We are now ready to formulate the vaccination game that determines $\mu_k$ for each degree-$k$ subpopulation. We denote $\boldsymbol{\mu} =  \left\lbrace \mu_{k}\right\rbrace_{\forall k}$ as the overall vaccine adoption level across the population. Moreover, we use superscription $i+$ to denote the vaccination decision of an individual $i$, and $i-$ to denote the collective decision of all others; e.g., a degree-$k$ individual $i$'s vaccine adoption level is denoted $\mu^{i+}_k$, and that of the rest of the population is $\boldsymbol{\mu}^{i-} = \left\lbrace \mu_{k'}\right\rbrace_{\forall k' \ne k} \bigcup \left\lbrace \mu^{i-}_k \right\rbrace$. Since the vaccine guarantees perfect immunity, the probability of a degree-$k$ individual $i$ eventually getting infected is the product of the conditional probability that it is infected given it is not vaccinated (NV) and the probability that it is NV. Given its own vaccine adoption level $\mu^{i+}_k$, and all others' decisions $\boldsymbol{\mu}^{i-}$, the infection probability of $i$ is:
\begin{equation*}
\begin{split}
& Pr\left\lbrace \mbox{Infected } \big| \mu^{i+}_k, \boldsymbol{\mu}^{i-}, k \right\rbrace \\
= &  Pr\left\lbrace \mbox{Infected } \big| \mbox{NV},  \mu^{i+}_k, \boldsymbol{\mu}^{i-} , k \right\rbrace Pr\left\lbrace \mbox{NV} \big| \mu^{i+}_k, \boldsymbol{\mu}^{i-}, k  \right\rbrace \\
= & Pr\left\lbrace \mbox{Infected } \big| \mbox{NV}, \boldsymbol{\mu}^{i-}, k \right\rbrace Pr\left\lbrace \mbox{NV} \big| \mu^{i+}_k, \boldsymbol{\mu}^{i-}, k  \right\rbrace, \forall k = 1, \dots, K.\end{split}
\end{equation*}
We assume $f(\cdot)$ is a function that captures the psychological penalty caused by vaccine-phobia. Given people's vaccine-phobia level, the penalty function $f$ is identical for every individual, strictly increasing w.r.t. $\mu_k$, and satisfying $f(0) = 0$.
We assume that the total expected penalty (or disutility) for a degree-$k$ node $i$ is a linear combination of its own psychological penalty and infection risk, as follows:
\begin{equation}
E^i_k\left[ \mu^{i+}_k, \boldsymbol{\mu}^{i-} \right] = \alpha_1 f(\mu^{i+}_k) + \alpha_2 Pr\left\lbrace \mbox{Infected } \big| \mu^{i+}_k, \boldsymbol{\mu}^{i-}, k \right\rbrace, \forall k = 1, \dots, K,
\end{equation}
where $\alpha_1, \alpha_2 \ge 0$ are weights.

In a large population, the vaccination decision of one individual does not significantly affect the overall vaccination coverage, i.e., $\boldsymbol{\mu}^{i-}$ is approximately the same for all $i$'s, and hence  approximated by $\boldsymbol{\mu}$. The expected penalty of $i$ can be written as:
\begin{equation}
\begin{split}
& E_k^i\left[ \mu^{i+}_k, \boldsymbol{\mu}^{i-} \right] \\ \approx & \alpha_1 f(\mu^{i+}_k) + \alpha_2  Pr\left\lbrace \mbox{Infected } \big| \mbox{NV}, \boldsymbol{\mu}, k \right\rbrace Pr\left\lbrace \mbox{NV} \big| \mu^{i+}_k, \boldsymbol{\mu}, k \right\rbrace, \forall  k = 1, \dots, K. \label{eq: expanded utility function}
\end{split}
\end{equation}

Since all not-vaccinated degree-$k$ nodes are probabilistically identical, the probability for one of them to get infected shall be simply the following ratio: 
\begin{equation}
Pr\left\lbrace \mbox{Infected } \big| \mbox{NV}, \boldsymbol{\mu}, k \right\rbrace  = \frac{R_k(\infty, \boldsymbol{\mu})}{R_k(\infty, \boldsymbol{\mu}) + S_k(\infty, \boldsymbol{\mu})}, \forall k = 1, \dots, K.
\end{equation}
Moreover, we need to find out the probability of $i$ not getting vaccinated, $Pr\left\lbrace \mbox{NV} \big| \mu^{i+}_k, \boldsymbol{\mu}, k \right\rbrace$. We first denote the probability of node $i$ being in status $X$ at time $t$ as $X_k^{i}(t)$, $\forall X \in \{S,I,R,V,A\}$. Conditional on that $i$ is susceptible at time $t$, the vaccination rate is $v_k^{i+}$, and the infection rate is the same with all other degree-$k$ susceptible nodes. Therefore, the dynamics of $S_k^{i+}$ follow a similar form as Eq. \eqref{eq: dot S}:
\begin{equation}
\dot{S}_k^{i} =\left( - v_k^{i+} - k r \frac{\phi}{\theta} \right)  S_k^i, \forall t \geq 0, k = 1, \dots, K.
\end{equation}
We easily obtain $S_k^{i}(t) = \theta^{k + \mu_{k}^{i+}}$ as a solution. The probability that $i$ eventually being vaccinated or activated is $\int_0^{\infty} v_{k}^{i+} S_k^{i} dt$, which leads to the following equation:
\begin{equation}
Pr\left\lbrace \mbox{NV} \big| \mu^{i+}_k, \boldsymbol{\mu}, k \right\rbrace = \frac{k + \mu_{k}^{i+} \left( \theta_{\infty} \left(\boldsymbol{\mu} \right) \right)^{k + \mu_{k}^{i+}} }{ k + \mu_{k}^{i+}}, \forall \mu^{i+}_k, \boldsymbol{\mu} \geq 0, k = 1,\dots, K. \label{eq: prob nv}
\end{equation}
Eq. \eqref{eq: prob nv} is a strictly decreasing and strongly convex function with respect to $\mu_k^{i+}$ for any $\theta_{\infty} \in (0,1)$, which can be verified from its first and second order derivatives. If $f(\mu_{k}^{i+})$ is convex with respect to $\mu_{k}^{i+}$, Eq. \eqref{eq: expanded utility function} is also strongly convex with respect to $\mu_{k}^{i+}$ for any $\alpha_1, \alpha_2 > 0$. In this case, there exists a unique $\mu_k^{i+*}$ that minimizes the expected penalty for $i$, and it is the \textit{best response} of individual $i$ given all others' decisions:
\begin{equation}
B_i\left( \boldsymbol{\mu}^{i-} \right)  = \mu_k^{i+*} = \argmin_{\mu_k^{i+}} E_k^i \left[\mu^{i+}_k, \boldsymbol{\mu} \right], \forall  k = 1, \dots, K. 
\end{equation}
These results naturally lead to the following proposition, which highlights the Nash Equilibrium among the subpopulation with the same degree.

\begin{proposition} \label{prop: same rate across same degree}
	For any fixed average vaccination rates of all other subpopulations, $\mu_{k'} r \phi/\theta, \forall k' \ne k$, and $\forall \alpha_1, \alpha_2 > 0$, if the penalty function for vaccination $f(x)$ is convex, the vaccination rates among all degree-$k$ nodes should be equal at Nash Equilibrium; i.e., $B_i \left( \boldsymbol{\mu}^{i-}\right) = B_j \left( \boldsymbol{\mu}^{j-}\right)$ if $i$ and $j$ both have degree $k, \forall k = 1, \dots, K$. 
\end{proposition}

Proposition \ref{prop: same rate across same degree} states that the optimal vaccination rate of nodes with the same degree should be equal under equilibrium. This is as well intuitive because nodes with the same degree are probabilistically identical.  The following proposition further speaks to the existence and uniqueness of such a solution.

\begin{proposition} \label{prop: equilibrium condition}
	(Sufficient condition) Suppose $f(x)$ is a convex function and Proposition \ref{prop: same rate across same degree} holds. The solution to following equations:
	\begin{align}
	& \alpha_1 \frac{\partial f(\mu^*_k) }{\partial \mu_k }  + \alpha_2 \frac{R_k(\infty, \boldsymbol{\mu}^*)}{R_k(\infty, \boldsymbol{\mu}^*) + S_k(\infty, \boldsymbol{\mu}^*)} \frac{\partial }{\partial \mu_k } \left( \frac{ k + \mu_k^* \theta^{k + \mu_k^*}_\infty}{ k + \mu_k^* }  \right)  = 0, k = 1, \dots, K, \label{eq: equilibrium conditions}\\
	& \boldsymbol{\mu}^* \geq 0 , \label{eq: non-neg}
	\end{align}
	denoted $\boldsymbol{\mu}^* = \left\lbrace \mu_{k}^* \right\rbrace_{\forall k} $, exists and characterizes the vaccination decisions among the entire population under Nash Equilibrium. Moreover, if $f'(0) = 0$ and $f$ is strongly convex, there always exists a unique $\boldsymbol{\mu}^*$ that satisfies Eq. \eqref{eq: equilibrium conditions} and \eqref{eq: non-neg}.
\end{proposition}
\begin{proof}
	Eq. \eqref{eq: equilibrium conditions} are simply the first order conditions of an individual's objective function, which come naturally with the equilibrium definition and Proposition \ref{prop: same rate across same degree}. To show the existence and uniqueness of the equilibrium, we first observe that for any $\boldsymbol{\mu} \geq 0$, and any fixed $\tilde{\theta}_\infty \in \left[ \theta_\infty\left( \mathbf{0}\right),1\right)$, where $\theta_\infty\left( \mathbf{0}\right)$ is the final state of $\theta$ under $\boldsymbol{\mu} = \mathbf{0}$, there always exists a unique solution $\tilde{\mu}_k$ to
	\begin{equation}
	\alpha_1 \frac{\partial f(\tilde{\mu}_k) }{\partial \mu_k }  + \alpha_2 \frac{k - k \tilde{\theta}_{\infty}^{k+\mu_k}}{k + \mu_{k} \tilde{\theta}^{k+\mu_k }_\infty} \frac{\partial }{\partial \mu_k } \left( \frac{ k + \tilde{\mu}_k \tilde{\theta}^{k + \tilde{\mu}_k}_\infty}{ k + \tilde{\mu}_k }  \right)  = 0, k = 1, \dots, K. \label{eq: first order condition at tilde theta}
	\end{equation} 
	This is because while $f'(\mu_k)$ is an strictly increasing function from 0 (at $\mu_k = 0$), the second term is also strictly increasing with a negative starting value at $\mu_k = 0$, and it approaches 0 when $\mu_k \rightarrow \infty$. Therefore, there always exists a unique solution to Eq. \eqref{eq: first order condition at tilde theta}, and we can write it as a function of $\tilde{\theta}_\infty$, namely $\tilde{\mu}_k \left( \tilde{\theta}_\infty \right), k = 1,\dots, K$. 
	
	Moreover, by taking partial derivatives on Eq. \eqref{eq: first order condition at tilde theta} over $\tilde{\theta}_\infty$, we see that $\tilde{\mu}_k $ should be continuous and strictly decreasing with $\tilde{\theta}_\infty$, $\forall k = 1,\dots, K$. Furthremore, $\theta_\infty(\boldsymbol{\mu})$ is continuous and monotonically increasing with respect to $\mu_k, \forall k = 1, \dots, K$. Therefore, there must exists a unique $\tilde{\theta}_\infty \in \left[ \theta_\infty\left( \mathbf{0}\right),1\right) $, such that $\theta_\infty \left( \boldsymbol{\mu}^* \right) = \tilde{\theta}_\infty$, where $\boldsymbol{\mu}^* = \left\lbrace \tilde{\mu}_k \left( \tilde{\theta}_\infty \right) \right\rbrace_{\forall k} $.
\end{proof}

Proposition \ref{prop: equilibrium condition} characterizes the vaccination game equilibrium condition using a simple first order condition, and proposes a sufficient condition for the existence and uniqueness of the equilibrium. The requirements that $f'(0) = 0$ and $f$ be strongly convex are essentially stating that the marginal psychological penalty of vaccination increases with the vaccine adoption level, which could be easily satisfied in many real-world cases.\footnote{The convexity of the penalty function is analogous to an individual's risk aversion behavior in general, which is represented by a concave utility function, or equivalently a convex disutility function.} 

It is not possible to directly obtain an analytical solution to the equilibrium, however, because $ \theta_{\infty}$ lacks an explicit form in terms of $\boldsymbol{\mu}^*$, and requires solving the non-linear equation Eq. \eqref{eq: theta_infty}. In the next subsection, we utilize Proposition \ref{prop: equilibrium condition} to design a heuristic algorithm that solves for the equilibrium point. 

\subsection{Solution Algorithm}\label{sec:algorithm}
The basic idea of the heuristic solution algorithm is to iteratively find the optimal vaccine adoption level for a constant value of $\theta_{\infty}$, and then update $\theta_{\infty}$ with the new vaccination rate. Given $\theta_{\infty}$ in each iteration, the problems are convex so the optimal vaccination rates are found via a bisection search; then $\theta_{\infty}$ is updated by solving Eq. \eqref{eq: theta_infty}. The algorithm is described as follows:
\par
\textbf{Algorithm 1}
\begin{itemize}[noitemsep, topsep = 0pt]
	\item \textbf{Step 0}: Set $n = 0$, $\boldsymbol{\mu}^{(n)} = \mathbf{0}$. Solve for $ \theta_\infty^{(n)} = \theta_\infty\left( \boldsymbol{\mu}^{(n)}\right) $ and $R_k\left( \infty, \boldsymbol{\mu}^{(n)}\right) $ with Eq. \eqref{eq: theta_infty} and \eqref{eq: special final Rk}, respectively, and solve for $\hat{\mu}_k$ such that $\alpha_1 f\left( \hat{\mu}_k \right)  = \alpha_2 R_k\left( \infty, \boldsymbol{\mu}^{(n)}\right)$, $\forall k$;
	\item \textbf{Step 1}: Use $0$ and $\hat{\mu}_k$ as the lower and upper bounds, respectively, and perform bisection search to find the solution $\mu^{(n+1)}_k$ to the following equations:
	\begin{equation}
	\alpha_1 \frac{\partial f\left(  \mu^{(n+1)}_k \right) }{\partial \mu_k }  + \alpha_2 \frac{k - k \left( \theta_{\infty}^{(n)}\right)^{k+\mu_k^{(n)}}}{k + \mu_{k}^{(n)} \left( \theta_\infty^{(n)}\right) ^{k+\mu_k^{(n)}}} \frac{\partial }{\partial \mu_k } \left( \frac{ k + \mu^{(n+1)}_k \left( \theta_\infty^{(n)} \right)^{k + \mu^{(n+1)}_k }}{ k + \mu^{(n+1)}_k } \right) = 0, k = 1, \dots, K. \label{eq: solution algorithm}
	\end{equation}
	\item \textbf{Step 2}: Solve for $\theta_\infty^{(n+1)}$ using Eq. \eqref{eq: theta_infty} with $\boldsymbol{\mu}^{(n+1)}$;
	\item \textbf{Step 3}: If any of the following conditions are satisfied, terminate and return $\boldsymbol{\mu}^{(n+1)}$; otherwise, let $n \leftarrow n + 1$ and go to Step 1. The termination conditions include:
	\begin{itemize}
		\item Maximum iteration number is reached; 
		\item $\left|\boldsymbol{\mu}^{(n+1)} - \boldsymbol{\mu}^{(n)} \right| \leq \epsilon$, where $\epsilon$ is a user-defined convergence threshold; and
		\item $\left|\theta_\infty^{(n+1)} - \theta_\infty^{(n)} \right| \leq \epsilon$.
	\end{itemize}
\end{itemize}

In each iteration, we need to solve $K$ nonlinear equations as suggested by Eq. \eqref{eq: solution algorithm}, which might appear formidable when $K$ is large. Luckily, existing research show that real-life contact networks are more suitably described by degree distributions without a heavy tail \citep{kossinets2006empirical, srinivasan2006analysis}, which means the probability of $k$ exceeding a finite large number can be neglected. Moreover, bisection search solves each one of the one-dimensional nonlinear equations fairly efficiently. 

Although theoretical convergence of the above algorithm is not guaranteed, it is not difficult to see that the algorithm terminates when the equilibrium point is reached. In fact, good convergence performances are observed through intensive experiments in many realistic Nash game contexts. In addition, we also found that the convergence and/or stability of the algorithm is rather insensitive to the step size. 

\section{Other Vaccination Schemes} \label{sec: other schemes}
In the dynamic vaccination game, we have implicitly assumed that people vaccinate themselves during the course of the epidemic outbreak (i.e., called \textit{delayed vaccination}), and people make different decisions based on degree heterogeneity (i.e., called \textit{heterogeneous vaccination}). Note, however, that in the real world, people may also choose vaccination prior to the onset of the outbreak (which will be called \textit{early vaccination}),  and/or enforce a vaccination rate indiscriminatively (e.g. vaccination programs enforced by a public health agency, which will be called \textit{homogeneous vaccination}). While it is expected that early vaccination should outperform its delayed counterpart, we aim to quantitively measure the impacts caused by the delay and in so doing provide useful insights for decision makers. 

To evaluate the impacts of these alternative vaccination options, we compare the following three schemes (I) early and homogeneous vaccination; i.e., prior to the disease outbreak, a fraction of $M = V \left( \infty \right)  + A\left( \infty \right)$ population are randomly vaccinated regardless of their degrees; (II) early and heterogeneous vaccination; i.e., prior to the disease outbreak, a $\left( A_k(\infty) + V_k(\infty) \right)/p_k$ fraction among the degree-$k$ population are vaccinated, $\forall k = 1,\dots , K$; and (III) delayed and homogeneous vaccination; i.e., during the epidemic propagation process, the population are vaccinated homogeneously following the same total rate as the vaccination game, i.e., there is a fraction of $\sum_k v_k(t) S_k(t) dt$ susceptible individuals vaccinated in the time interval $[t, t + dt)$, $\forall t \geq 0$. 
In each of these schemes, the total number of individuals eventually receiving vaccination is set to be the same as that of the dynamic vaccination game in the previous section. However, the disease propagation process and the final epidemic size may be different. 

We first follow \cite{newman2002spread} to compute the final epidemic size of Scheme I. The vaccination scheme is equivalent to randomly removing a fraction of $M$ nodes from the network before the disease starts to propagate. After removing these nodes, not only the total number of nodes, but also the degree distribution in the remaining network changes. To find the new degree distribution in the remaining network, we note that randomly removing $M$ fraction of nodes from the network is equivalent to removing $M$ fraction of the stubs. The probability of a degree-$k$ node having $m$ stubs connected to a remaining stub is given by a binomial distribution with probability $\binom{k}{m} (1-M)^m M^{k-m}$. Then the degree generating function of the remaining nodes, denoted as $g_1(x)$, is:
\begin{equation}
\begin{split}
g_1(x) & = \sum_{m} x^m \sum_{k} p_k \binom{k}{m} (1-M)^m M^{k-m} \\
& = \sum_{k} p_k \sum_{m} \binom{k}{m} \left( x(1-M) \right) ^m M^{k-m} \\
& = \sum_{k} p_k \left( x(1-M) + M \right)^k \\
& = g\left( x(1-M) + M\right).
\end{split} \label{eq: g1x}
\end{equation}
The same result is also derived in \cite{buldyrev2010catastrophic}.

When the infectious disease starts to propagate in the network, there are two possibilities: (i) there are only local outbreaks such that the epidemic size does not scale with the population; and (ii) there is a pandemic outbreak such that a single disease seed leads to a giant connected subnetwork of infected nodes. \cite{newman2002spread} presents the critical condition on the disease transmissibility to determine whether a pandemic outbreak could happen; i.e., a pandemic outbreak may occur only when $T > T_c$, where
\begin{equation}
T_c = \frac{g_1'(1)}{g_1''(1)} = \frac{g'(1)}{ \left( 1 - M \right) g''(1)}. \label{eq: Tc}
\end{equation}
If there is no pandemic outbreak, the average number of infected nodes, is given by:
\begin{equation}
s_0  = 1 + \frac{T g_1'(1)}{1 - T g_1''(1) / g_1'(1)} = 1 + \frac{T \left( 1 - M \right) g'(1) }{ 1  - T  \left( 1 - M \right) g''(1) / g'(1) }. \label{eq: s0}
\end{equation}
Otherwise, if the pandemic outbreak occurs, the fraction of the giant infected subnetwork out of the unvaccinated population is given by $1 - g_1(1 + (y-1) T)$, where $y$ is the solution to the following self-consistency relation:
\begin{equation}
y = \frac{1}{g_1'(1)} g_1'(1 + (y-1)T).
\end{equation}
Then the following equation yields the pandemic size in terms of fraction among the original population:
\begin{equation}
s_1 = (1 - M) \left( 1 - g_1\left( 1 + (y-1) T\right) \right) . \label{eq: final epidemic size vac before epidemics}
\end{equation}
Moreover, the average outbreak size among the nodes not belong to be giant subnetwork is still given by Eq. \eqref{eq: s0}.

For Scheme II, the population is heterogeneously vaccinated. Although there exists a similar analytical approach in finding the final epidemic size, e.g. as in \cite{huang2011robustness}, it requires the vaccination probability of nodes to follow a certain special form. Therefore, in this paper, we will simply use stochastic simulations to compute the system evolution and the final epidemic size. 
Finally, we note that Scheme III is a special case of dynamic vaccination. Hence, the outcome of Scheme III can be easily computed using the developed SIRVA dynamics model.

\section{Numerical Results}\label{sec:numerical}
In this section, we first validate the proposed SIRVA dynamic model (Eqs. \eqref{eq: dot theta natural}, \eqref{eq: dot R natural}, and \eqref{eq: fixed popualtion size} - \eqref{eq: A}) as well as the derived final epidemic size (Eqs. \eqref{eq: P} and \eqref{eq: theta_infty}), by comparing these analytical formulas with agent-based simulation outcomes. Then we present the results of the vaccination game and evaluate the impacts of delayed and homogeneous vaccinations. Sensitivity analysis is then conducted to reveal interesting insights. 

The benchmark case involves a population with a Poisson degree distribution and a hypothetical epidemic event with moderate disease transmissibility. Poisson distribution is chosen because it does not possess heavy-tailed properties and has been widely recognized in the existing literature \citep{newman2001random, barthelemy2005dynamical, volz2008sir}. The average degree $g'(1) = 7$, with a cutoff $K = 22$ (such that $P\left( k \geq K\right) \leq 10^{-5}$). The disease parameters are as follows: the infection rate $r = 0.01$, the removal rate $u = 0.01$, and consequently the disease transmissibility $T = 0.5$. Moreover, the initial condition at time $t=0$ is set as follows: $\theta(0) = 1.0$, and $\phi(0) = 0.001$.

\subsection{Model Validation} \label{sec:validation}
In the agent-based stochastic simulations, configuration model random networks are built following \cite{molloy1995critical}, where open edges are generated and assigned to the nodes following the degree distribution, and then randomly paired up with other open edges. The number of nodes (i.e., the population size) is set to be $10^4$; in every simulation, initial disease seeds are randomly located at 10 of these nodes. 

We perform a total of 50 simulations, each with a new realization of network topology and initial disease seeds. In each discretized time step (e.g., a day), disease transmission, node removal, and node vaccination, and node activation are randomly simulated with corresponding probabilities.  
The vaccination rate follows Eq. \eqref{eq: special v}, with $\mu_k = 0.4 k, \forall k$. In the simulations, since $\theta$ and $\phi$ are virtual quantities which cannot be measured directly, we use the following approximation of $\phi/\theta$ to compute $v_k$:
\begin{equation}
r \frac{\phi}{\theta} \approx \frac{\Delta I}{\sum_k k S_k},
\end{equation}
where $\Delta I$ is the fraction of recent infections. 

\begin{figure}[t]
	\centering
	\includegraphics[width=0.9\textwidth]{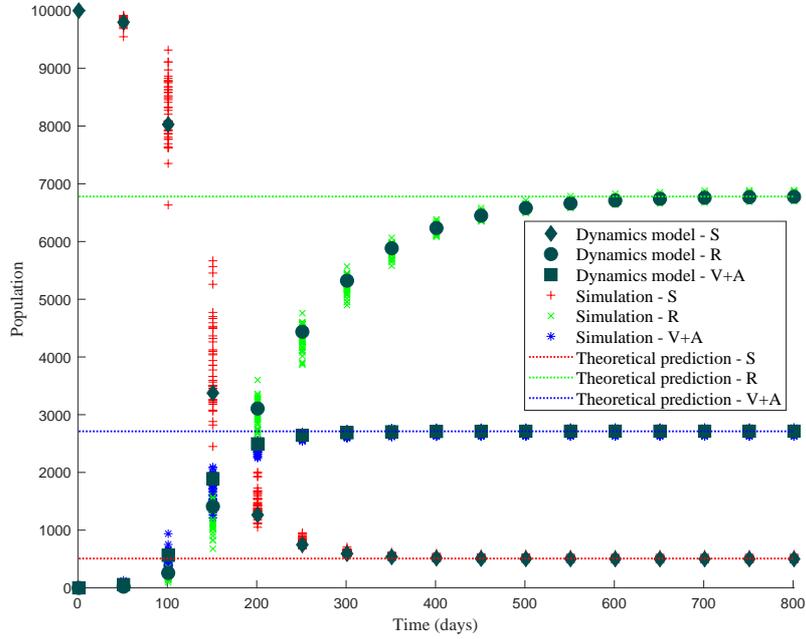}
	\caption{Compare simulation with system dynamic equations and theoretical final states.}
	\label{fig: simulation}
\end{figure}

The simulation results can be compared with those from the system dynamics equations, as well as the analytical formulas \eqref{eq: theta_infty} and \eqref{eq: P}. The colored dots in Figure \ref{fig: simulation} show results from the 50 stochastic simulations, sampled every 50 days. While these simulated samples show some variations in the early stages of the simulation, the final state of the epidemic converges fairly well to the same value across all simulations: the total fraction of immunized population is 0.2687, that of the remaining susceptible is 0.0510, and that of the removed is 0.6803.  The results from system dynamics equation coincide well with the simulation results: the final epidemic size $R(\infty)$ computed from Eq.\eqref{eq: dot theta natural}, \eqref{eq: dot R natural}, and \eqref{eq: fixed popualtion size}-\eqref{eq: A} is 0.6786, which is only 0.3\% different from that of the simulations, and the fraction of remaining susceptible is 0.0501. Despite the stochasticity of the simulations, the SIRVA dynamic model and the simulations agree very well. In fact, we believe the relative difference will further reduce if the population size increases. Furthermore, Eqs. \eqref{eq: theta_infty} and \eqref{eq: P} can be used directly to compute $\theta_\infty = 0.6738$ and the final epidemic size equals 0.6780. The relative error between the theoretical final epidemic size and that of the system equations is negligible regardless of initial condition of $\phi$. These results give us reasonable confidence to use the established systems dynamics models and the derived analytical results to study vaccination games. 

\subsection{Vaccination Game Results}\label{sec: vac game res}
We proceed to present the results of a hypothetical vaccination game with the same population and disease. For convenience, we set the weight factors $\alpha_1=10^{-4}$ and $\alpha_2=1$. The penalty for vaccination is set to be a polynomial function in the form of $f(x) = x^b$ where $b = 2$. As such, $f(x)$ is a monotonically increasing function, and it satisfies the equilibrium uniqueness requirements in Proposition \ref{prop: equilibrium condition}, thus the vaccination game has a unique equilibrium solution. The heuristic algorithm terminates after 11 iterations when $\epsilon$ reaches $1 \times 10^{-4}$.

\begin{figure}[h] 
	\centering
	\begin{subfigure}[b]{0.48\textwidth}
		\includegraphics[width=\textwidth]{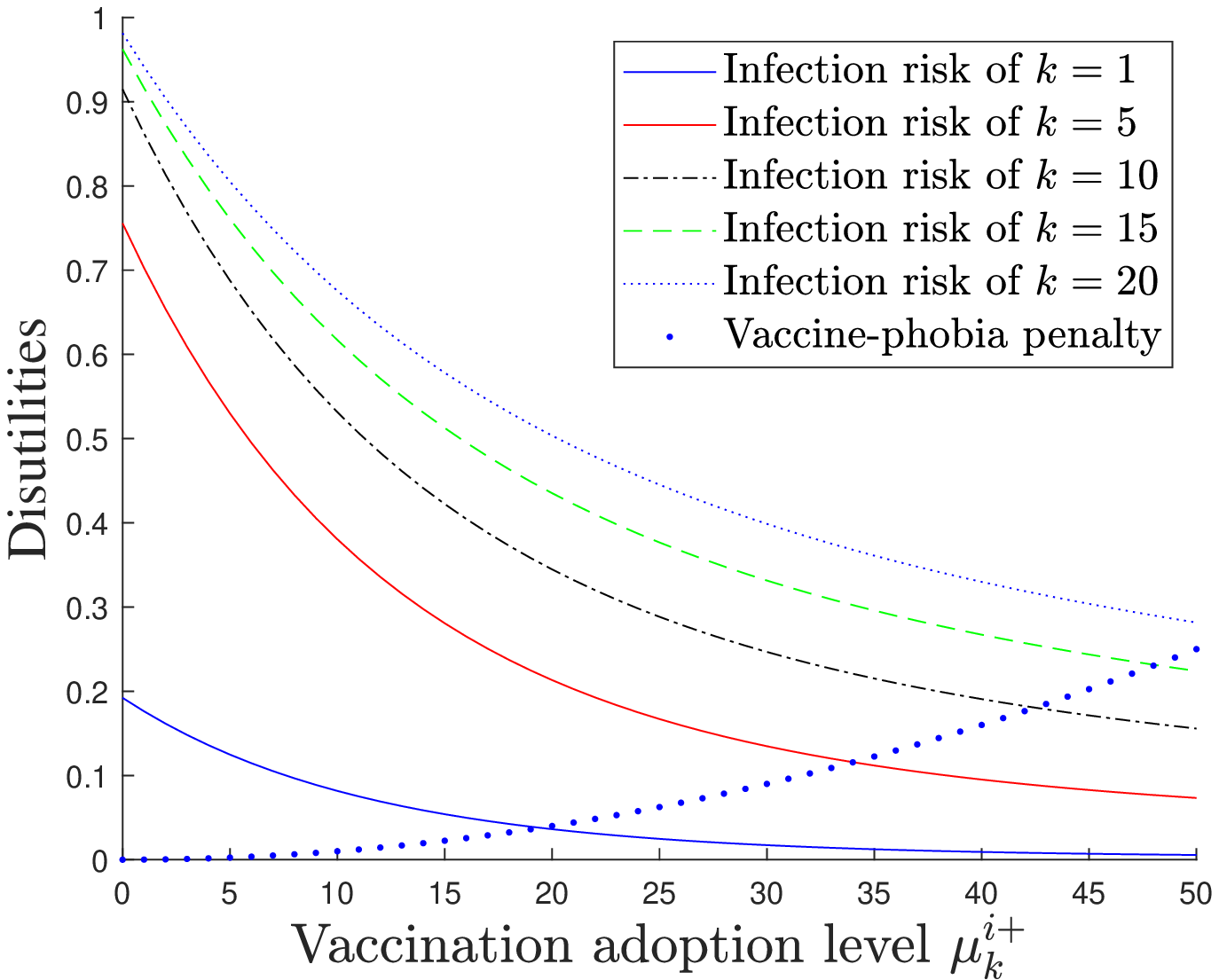}
		\caption{Disutilities w.r.t. $\mu_k^{i+}$}
		\label{fig: disutilities}
	\end{subfigure}
	~ 
	\begin{subfigure}[b]{0.48\textwidth}
		\includegraphics[width=\textwidth]{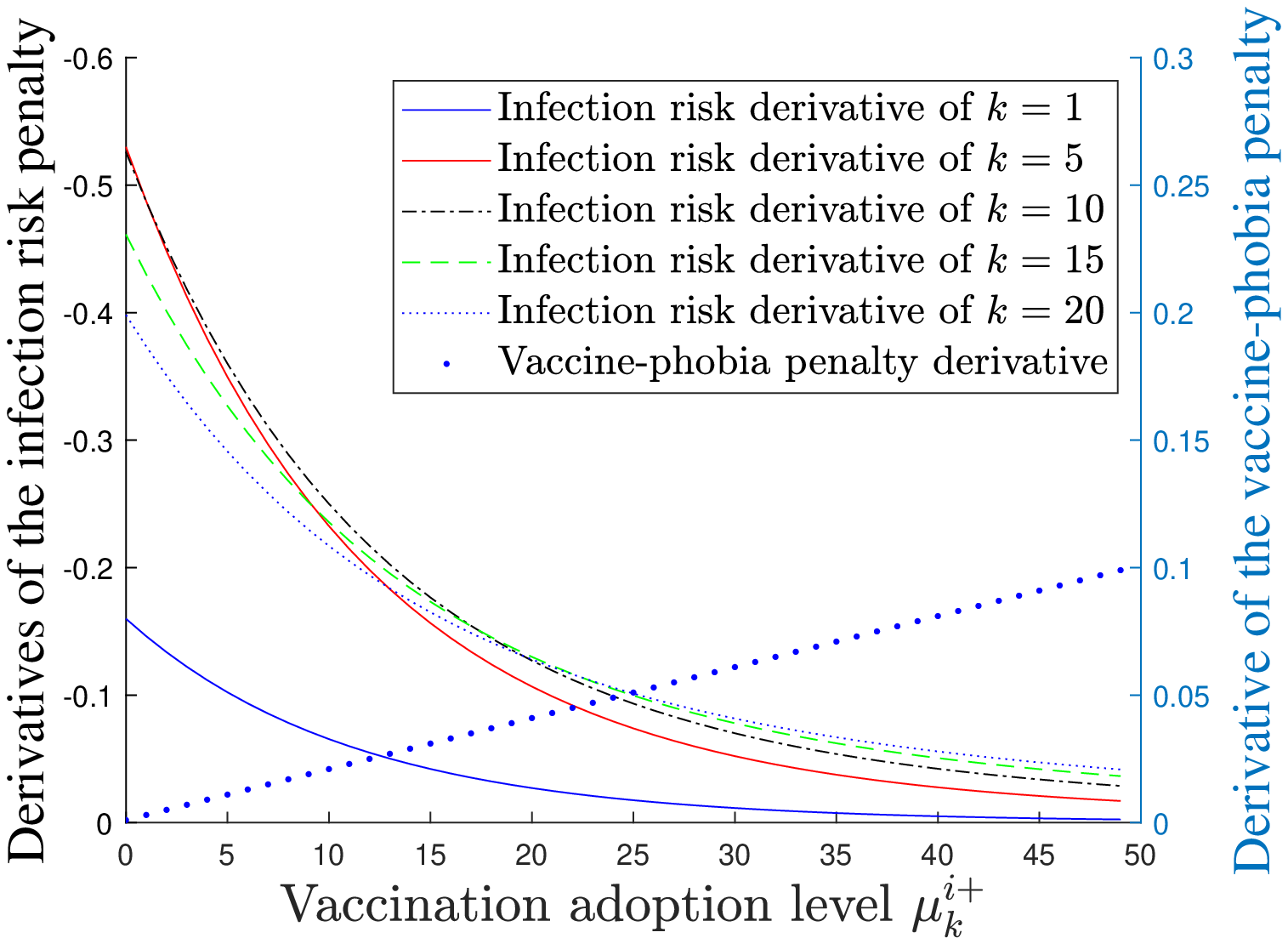}
		\caption{First order derivatives of disutilities w.r.t. $\mu_k^{i+}$}
		\label{fig: derivatives}
	\end{subfigure}
	\caption{Validation of the vaccination game results}\label{fig: validation}
\end{figure}

Figure \ref{fig: validation} presents an individual's disutilities due to vaccine-phobia and infection risk, as well as the first order derivatives of these disutilities at the equilibrium. As expected, the infection risk penalties are strictly decreasing functions of $\mu_k^{i+}$ for all $k$'s, and that of vaccine-phobia is a strongly convex and increasing function; see Figure \ref{fig: disutilities}. Moreover, Figure \ref{fig: derivatives} shows the first order derivatives of these disutilities. For each $k$, the intersection of the vaccination derivative curve and that of the infection is exactly the solution to the equilibrium condition Eq. \eqref{eq: equilibrium conditions}.

\begin{figure}[h] 
	\centering
	\begin{subfigure}[b]{0.48\textwidth}
		\includegraphics[width=\textwidth]{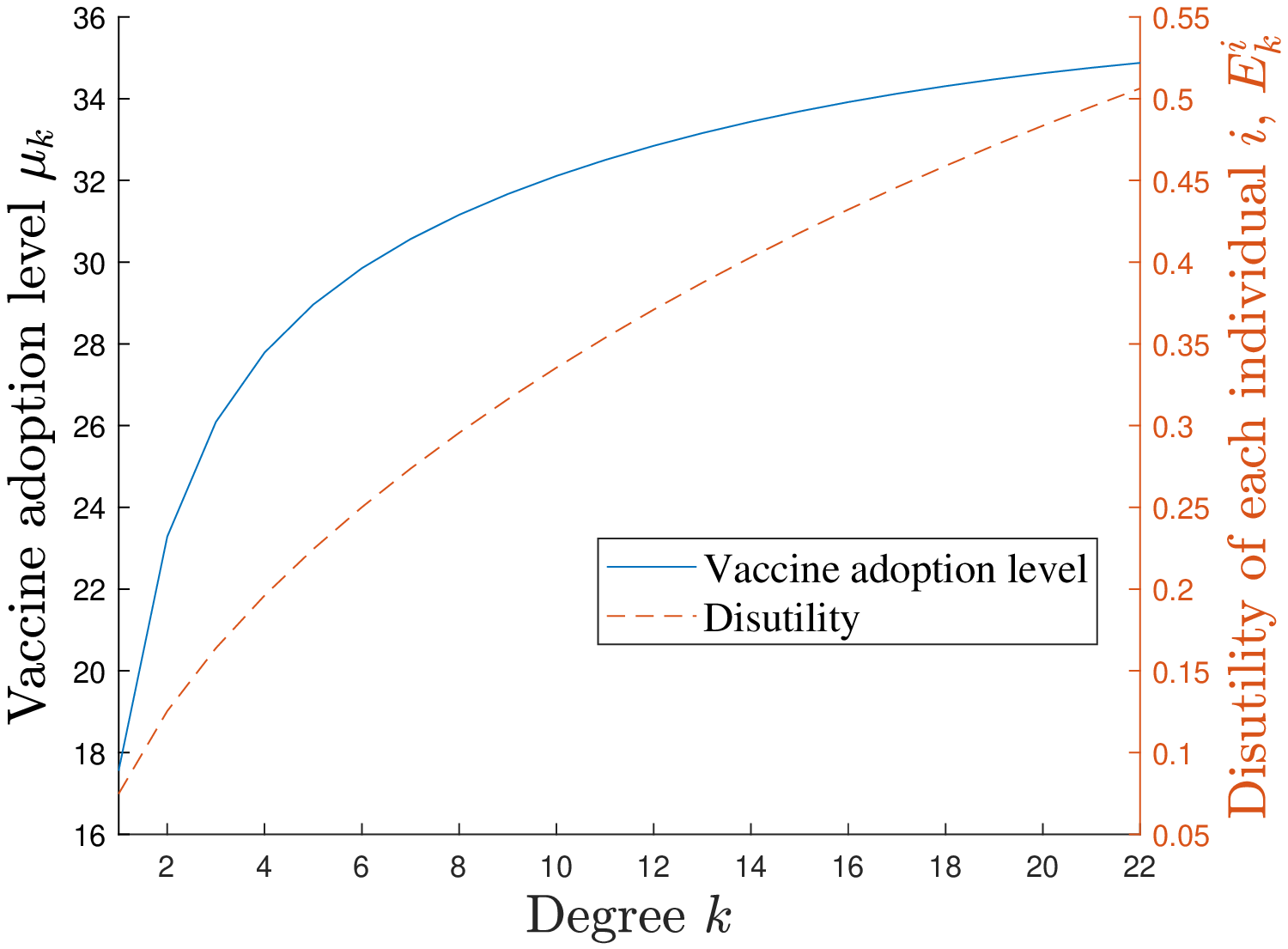}
		\caption{Vaccination rate decisions}
		\label{fig: uk}
	\end{subfigure}
	~ 
	\begin{subfigure}[b]{0.48\textwidth}
		\includegraphics[width=\textwidth]{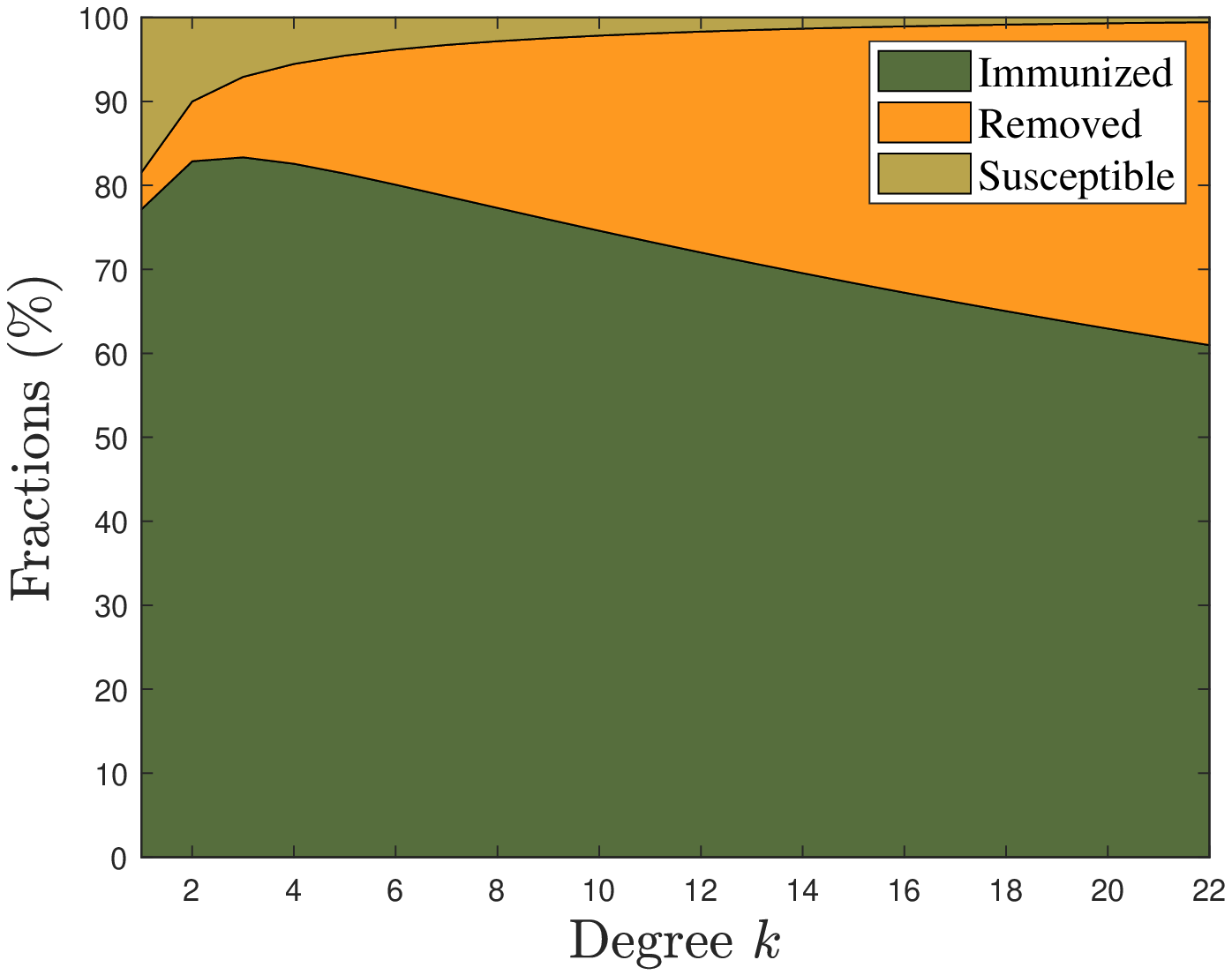}
		\caption{Final states}
		\label{fig: final states}
	\end{subfigure}
	\caption{Vaccination game equilibrium results}\label{fig: vaccination games}
\end{figure}

The equilibrium solutions are presented in Figure \ref{fig: vaccination games}. Figure \ref{fig: uk} shows that the equilibrium vaccination adoption level increases with the node degree; meanwhile, the total disutilities are strictly increasing with $k$. Figure \ref{fig: final states} shows the final state compositions of nodes with different degrees. Interestingly, even though the equilibrium vaccination rate strictly increases with $k$, the probability for a node to eventually get vaccinated first increases but then declines. This is because, the vaccination rate $v_k$ and infection rate are both proportional to $r \phi / \theta$, whereas the multiplicative factor equals $\mu_k$ for vaccination and $k$ for infection. Thus the infection rate grows proportionally to $k$, yet the vaccination rate grows sub-linearly as shown in Figure \ref{fig: uk}. Consequently, the ratio of vaccination rate to infection rate declines as the node degree increases. Implementing the vaccination decisions at the equilibrium in the systems dynamics model \eqref{eq: dot theta natural}, \eqref{eq: dot R natural}, and \eqref{eq: fixed popualtion size}-\eqref{eq: A}, the final epidemic size is 0.1770, with $\theta_\infty = 0.9124$, and a total fraction of 0.7866 of the population are eventually vaccinated. Compared to the results in Sec. \ref{sec:validation} where a relatively low vaccinate rate is adopted, we find that the disease propagation has been significantly mitigated as a result of the vaccination game.

\subsection{Delayed and Homogeneous Vaccinations}
Following Scheme I and considering vaccinating the same fraction 0.7866 of the population (as in the vaccination game) indiscriminatingly before the epidemic outbreak, we find that $T_c = 4.69 > T$ from Eq. \eqref{eq: Tc}. There will only be several endemic infections, and the average outbreak size is $1.8$ from Eq. \eqref{eq: s0}. This means that each disease seed is expected to infect only 0.8 additional susceptible individual during the entire outbreak. Since there are 10 initial disease seeds, and in the large population their neighbors are highly unlikely to overlap, the total infected fraction is expected to be approximately 0.0018.  Scheme II requires that the population be vaccinated heterogeneously based on their degrees prior to the epidemic outbreak. The mean infection fraction over 50 random simulations is 0.0051. 

The comparison between early and delayed vaccinations is consistent with our expectation. In Schemes I and II, the infection fractions are 99.0\% and 97.1\% smaller than that of the vaccination game case, respectively. This shows that, with preventive vaccination ahead of time, the pandemic outbreak is prevented from happening; in contrast, delaying the vaccination allows devastating outbreak to take place and infect a large fraction of the population. This is intuitive -- the higher-degree individuals are more easily infected and then pose larger threats to others once infected, thus vaccinating these individuals before the outbreak can greatly counteract the disease propagation. Furthermore, if we compare Schemes I and II, we find that non-discriminative vaccination is helpful if done prior to epidemic outbreak. Scheme I outperforms Scheme II because the high-degree nodes have higher vaccination probabilities in Scheme I, as implied by Figure \ref{fig: final states}. 

Scheme III, with both delayed and homogeneous vaccination, yields a final epidemic size of 0.1805. Comparing Scheme III and the dynamic vaccination game from Section \ref{sec: vac game res}, we observe that homogeneous vaccination is slightly less effective than heterogeneous vaccination if delayed. This is in sharp contrast to the comparison between Schemes I and II, mainly because the vaccination rate (instead of vaccination probability) of the high degree nodes are higher in the vaccination game as suggested by Figure \ref{fig: uk}, which helps protect other susceptible individuals. 

These results are interesting and insightful. The coverage and timing of vaccination for the high degree nodes almost dictate the outcome of a disease outbreak. As discussed in Section \ ref{sec: vac game res}, in a vaccination game people with more social connections tend to take vaccines sooner. Therefore, it would be more effective to encourage voluntary vaccine uptake than blindly vaccinating the whole population once the epidemic outbreak has started. This finding would be particularly helpful when vaccines are scarce during the course of an epidemic outbreak. However, despite their high vaccine adoption level, the high degree nodes are likely to be quickly infected before they can be vaccinated, as we have also seen in Section \ ref{sec: vac game res}. This facilitates disease propagation and leads to a much more devastating spread of disease. Therefore, it is important to vaccinate as much high degree nodes as possible prior to the outbreak to provide them with timely protection. In this case, non-discriminative vaccination is not a bad option – compared to the voluntary (but delayed) vaccination scheme, the high degree nodes now have a better chance of getting vaccinated in time. Targeting individuals with many social connections would greatly cut the possibility of disease propagation in both early and delayed vaccination schemes. Nevertheless, we shall also notice that such a measure is often not practical and may raise equity issues. 

\subsection{Sensitivity Analysis}
We conduct additional experiments to investigate the outcome of vaccination games under different disease-related, network-related, and penalty function related parameters. In all these tested cases, the equilibrium solution was found in no more than 15 iterations. For convenience, we call the vaccination game in Sec. \ref{sec: vac game res} as Scenario 0. 

We first change the degree distribution of the population, by considering the following cases: Scenarios (i) and (ii) change the average degree and the degree cutoff of the Poisson distribution, to $g'(1) = 4$ and $K = 16$, and $g'(1) = 10$ and $K = 27$, respectively; Scenario (iii) changes the degree distribution to a uniform distribution with the same average degree $g'(1) = 7$; then Scenarios (iv) and (v) consider uniform distributions with the same average degrees as Scenarios (i) and (ii), respectively. The vaccination rate decisions of all these scenarios are presented in Figure \ref{fig: compare dd}.
\begin{figure}[h!]
	\centering
	\includegraphics[width=0.8\textwidth]{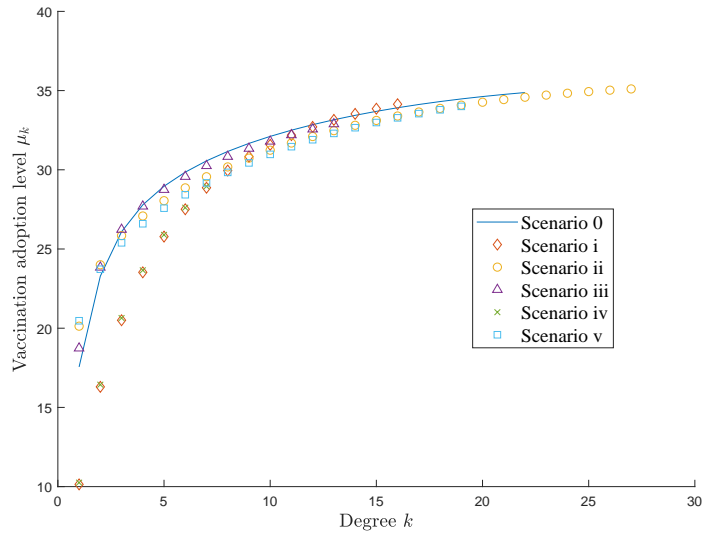}
	\caption{Compare vaccination decisions under different degree distributions.}
	\label{fig: compare dd}
\end{figure}

Interestingly, if the degree distributions, no matter Poisson or uniform, share the same average degree, they would lead to similar vaccination decisions for individuals of the same degrees. As the figure shows, Scenarios in the followings pairs, 0 and (iii), (i) and (iv), and (ii) and (v), all show exceptional agreement with each other. This suggests that the vaccine adoption level may not be affected by the shape of the degree distribution as long as they have the same average degree. However, it is worth mentioning that same vaccination decisions do not necessarily lead to the same epidemic outbreak outcome, including people's disutilities and final epidemic sizes, as the degree distributions still affects the disease propagation process. For instance, the final epidemic sizes in these scenarios are 0.1086, 0.2368, 0.1755, 0.1079, and 0.2287, respectively.

Moreover, the relationship between the average degree in the population and the equilibrium vaccination decisions is somewhat more complex. Figure \ref{fig: compare dd} shows that for low degree nodes (e.g., $k<5$), a higher average degree leads to a higher vaccination rate; however, the relationship is exactly the opposite for the high degree nodes (e.g, $k\geq15$). This phenomenon is probably due to two contradicting mechanisms. First, it is straightforward that high network connectivity increases the infection risk and thus encourages a high vaccination rate. However, with higher network connectivity, high-degree nodes will be infected so much faster, that in order to even slightly reduce their infection risk, a remarkably high vaccination adoption level is required. This induces high psychological suffering to a level that they would simply give up vaccination. A similar phenomenon can be observed in a later example.

Next, we examine the impact of the disease-related parameters $r$ and $u$. Note that in all the equations used to solve for the equilibrium, these two parameters appear together in the form of the disease transmissibility $T$. Therefore, we only need to tune the value of $T$. Since in Scenario (0) the transmissibility $T$ takes a moderate value of $0.5$, we consider three additional cases: Scenario (vi) and (vii) where $T = 0.3$ and $0.7$, respectively; and an extreme case Scenario (viii) where $T = 1$ (i.e., the removal rate $u = 0$). 

The results are presented in Figure \ref{fig: compare trans}. Again, the influence of disease transmissibility on people's vaccination decisions appears to be highly complex. The impact seems greater for low-degree nodes but becomes negligible for high-degree nodes. Furthermore, the final epidemic sizes of these three cases are 0.1611, 0.1873, and 0.1931, respectively, suggesting that disease transmissibility does not have a significant impact on the outcome of the epidemic in a vaccination game. This is probably because the overall vaccination coverage is already sufficiently high.

\begin{figure}[h]
	\centering
	\includegraphics[width=0.8\textwidth]{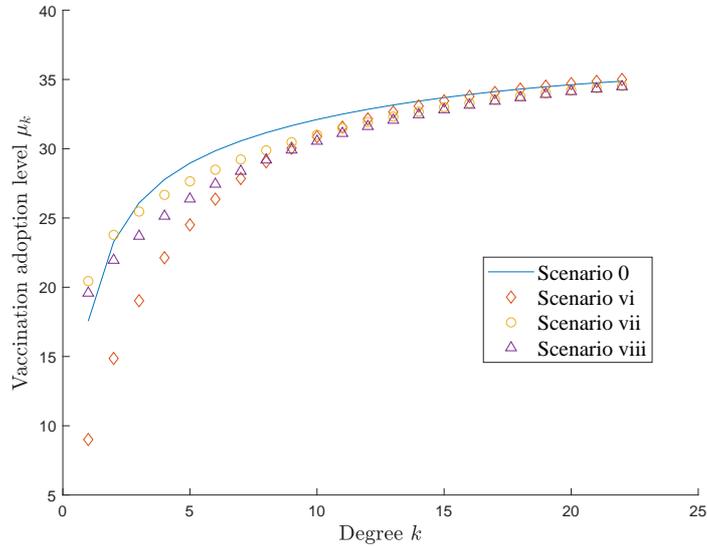}
	\caption{Compare vaccination decisions with different disease transmissibilities.}
	\label{fig: compare trans}
\end{figure}

Finally, we investigate the impact of the weight factors and the penalty function parameter. This is to examine the influence of vaccine-phobia as compared to infection penalties. Different combinations of the vaccination penalty parameters $\alpha_1$ and $b$ are tested, with $\alpha_2$ fixed to be 1, and the consequential infection and vaccination fraction, i.e. $R \left( \infty \right)$ and $V \left(  \infty \right)  + A \left( \infty \right) $ are presented, respectively, in Table \ref{table: sens analysis}. The results show, as expected, that people's vaccination decisions and the final outcome of the epidemic are largely dependent on the trade-off between infection risk and the vaccination penalty function. It suggests that vaccine-phobia has a huge impact on population vaccination decisions, and therefore the outcome of the epidemic outbreak.
\begin{table}[!htb]
	\centering
	\resizebox{\columnwidth}{!}{%
		\begin{tabular}{l | ccc }
			\hline
			Cases                  & $\alpha_1 = 1 \times 10^{-5} , b=1.5$ & $ \alpha_1 = 1 \times 10^{-4}, b=1.5$  & $\alpha_1=1 \times 10^{-3}, b=1.5$\\ \cline{1-1}
			Infection fraction 	& 0.0263         & 0.0736         & 0.2009   \\ \cline{1-1}
			Vaccination fraction   & 0.9201   & 0.8805          & 0.7596        \\
			\hline
			Cases                  & $\alpha_1= 1 \times 10^{-5}, b=2.0$ & \cellcolor{gray} $\alpha_1 = 1 \times 10^{-4}, b=2.0$ &  $\alpha_1=1 \times 10^{-3}, b=2.0$ \\ \cline{1-1}
			Infection fraction & 0.0790          & \cellcolor{gray} 0.1767          &  0.3646         \\ \cline{1-1}
			Vaccination fraction    & 0.8792         & \cellcolor{gray} 0.7851         &  0.6024                 \\
			\hline
			Cases                  & $\alpha_1 = 1 \times 10^{-5}, b=2.5$ & $\alpha_1 = 1 \times 10^{-4}, b=2.5$ & $\alpha_1=1 \times 10^{-3}, b=2.5$  \\ \cline{1-1}
			Infection fraction & 0.1594         & 0.2950         & 0.4958             \\ \cline{1-1}
			Vaccination fraction  & 0.8033          & 0.6711           & 0.4742    \\
			\hline
		\end{tabular}%
	}
	\caption{Sensitivity analysis on the penalty function parameters.}\label{table: sens analysis}
\end{table}

Moreover, a fairly counter-intuitive phenomenon is observed when the weight of vaccination penalty is relatively high. Taking the case with $\alpha_1=1 \times 10^{-3}$ and $b=2.5$ for instance, the vaccine adoption level of the population is shown in Figure \ref{fig: uk high vac pen}. High-degree nodes (e.g., $k\geq 10$) tend to give up on vaccinating themselves, as their vaccine adoption level is even lower than those with fewer degrees. Figure \ref{fig: derivatives high vac pen} demonstrates that at the equilibrium, for the high degree nodes, the marginal benefit in reducing the infection risk intersects with the marginal vaccine-phobia penalty at a low vaccine adoption level. This suggests it is difficult for them to further reduce their infection risk without inducing a much higher vaccination penalty. This implies that when the vaccine-phobia outweighs the penalty of infection, those who can be easily infected would rather take the risk of infection. In this case, non-differential vaccination enforced by the public health agency may work better than the largely voluntary vaccination game.

\begin{figure}[h] 
	\centering
	\begin{subfigure}[b]{0.48\textwidth}
		\includegraphics[width=\textwidth]{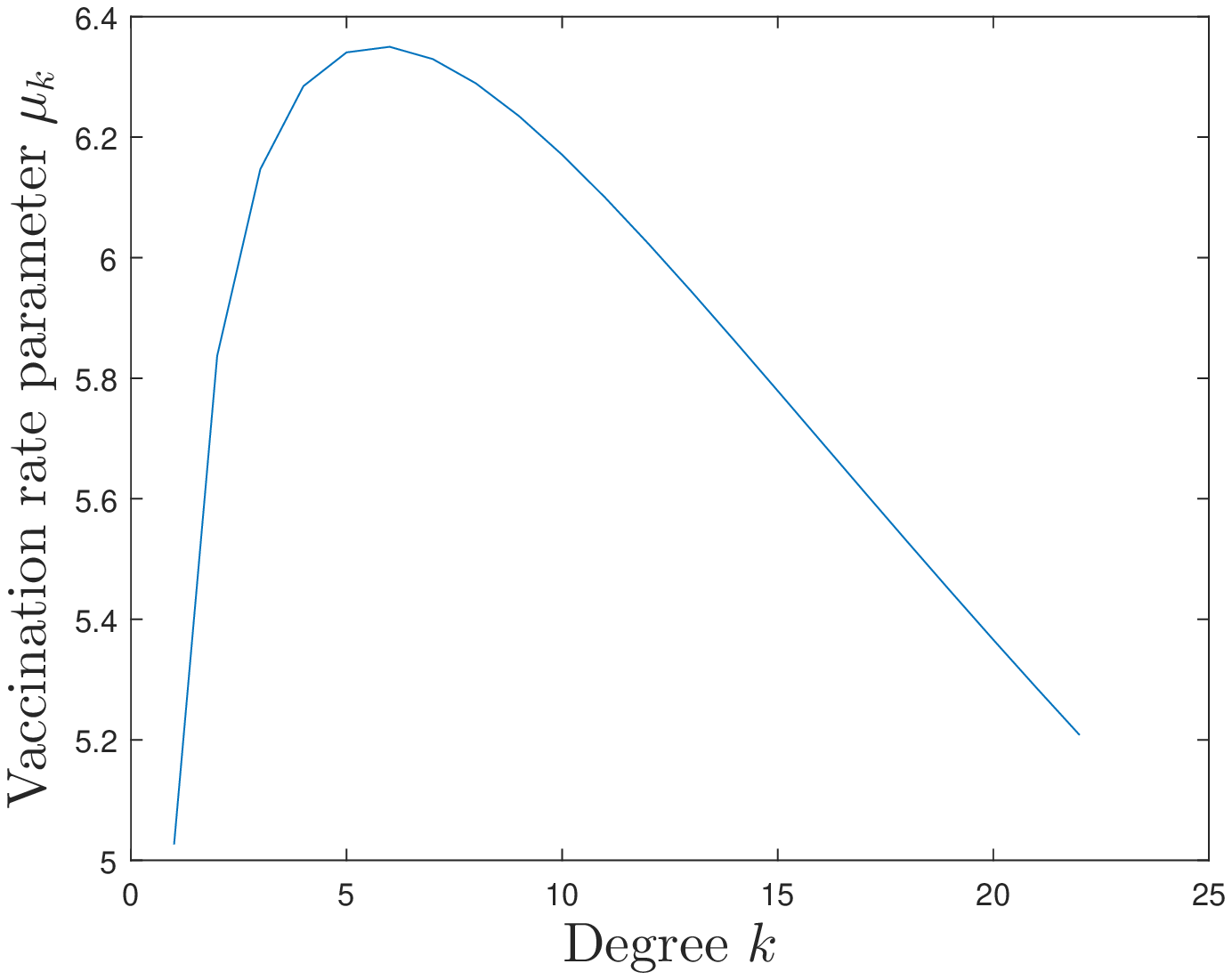}
		\caption{Vaccination rate decisions}
		\label{fig: uk high vac pen}
	\end{subfigure}
	~ 
	\begin{subfigure}[b]{0.48\textwidth}
		\includegraphics[width=\textwidth]{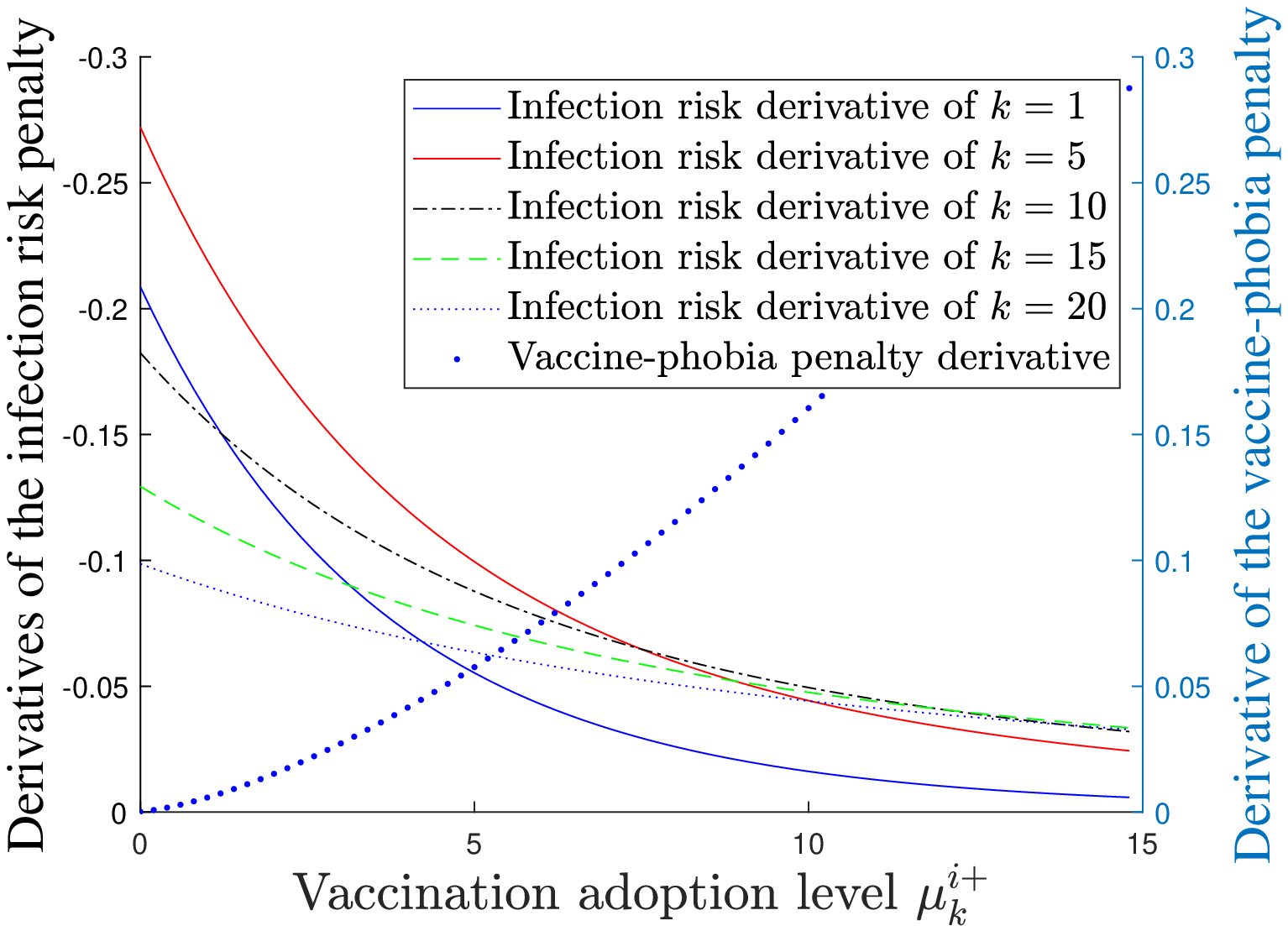}
		\caption{First order derivatives of disutilities w.r.t. $\mu_k^{i+}$}
		\label{fig: derivatives high vac pen}
	\end{subfigure}
	\caption{Vaccination decisions and equilibrium conditions of the high vaccination penalty scenario}\label{fig: high van pen}
\end{figure}

In another angle, we see that in a vaccination game, people's vaccination decisions are mostly affected by their relative perception of infection risk over vaccine-phobia, instead of the population connectivity or the disease transmissibility. Therefore, if a public health agency wishes to mitigate the propagation of an infectious disease by encouraging voluntary vaccine uptake among the population, it is important to alleviate people's fear towards vaccines and to properly inform them about the seriousness of the disease infection. Reducing population connectivity and disease transmissibility using other disease control approaches (e.g., quarantine) would not greatly affect people's vaccination behaviors, but they can also mitigate the propagation of the disease and should be considered along the course of disease spread as well.

\section{Conclusion}\label{sec:Conclusion}
This work investigates a dynamic vaccination game among a heterogeneous mixing population during an epidemic outbreak. Individual disutility includes the psychological penalty due to vaccine-phobia, as well as the expected penalty from potential infection risks. To this end, we first develop a generic SIRVA dynamic model to capture the epidemic dynamics incorporating the effect of vaccination. Then closed-form predictions of the final states of the dynamical system are derived, including the final epidemic size and the total vaccination coverage. The epidemic dynamics model and the derived final states are validated via agent-based stochastic simulations. Based on these results, a mathematical program that describes a Nash-type vaccination game is established. Several properties of Nash game are revealed through analytical derivations. A heuristic approach that solves for the equilibrium is then proposed, which is found to be quite efficient in our numerical experiments. Finally, a case study is performed to demonstrate the applicability of the vaccination game framework and the solution approach. Extensive numerical experiments are conducted to reveal the key factors that influence people's vaccination decisions and the outcome of the epidemic outbreak under equilibrium. Below are some interesting findings:
\begin{enumerate}[noitemsep,nolistsep]
	\item The final outcome of an epidemic outbreak is largely affected by the vaccination coverage and timing of high-degree individuals; therefore, delaying vaccination has a devastating impact on the final outcome of the epidemic; the performance of homogeneous vaccination (such as an enforced universal vaccination program) is more complex, depending on both the vaccination timing and people's vaccine-phobia level.
	\item The average degree among the population has a complex but quite limited impact on people's vaccination decisions, and the shape of the degree distribution across individuals does not affect their decisions; moreover, high disease transmissibility only encourages low-degree individuals to vaccinate sooner, but cannot make high-degree nodes to further increase their vaccination adoption level.
	\item The most critical factors that affect people's vaccination decisions is the relative weight of vaccination penalty as compared to the infection risk; if people have high vaccine-phobia or if the penalty of infection is low, high-degree nodes are observed to give up on vaccination; as a result, the final outcome of the epidemic outbreak is greatly affected by people's perception of risks.
\end{enumerate}

These findings suggest that it would always be more effective to increase vaccination coverage among high-degree individuals in a timely fashion to mitigate disease propagation. However, in reality targeted vaccination on high-degree nodes is not always practical. If vaccination can be done prior to the onset of a disease outbreak, an enforced vaccination program that non-discriminatively vaccinates the population could be very effective. If vaccination has to be delayed (e.g., because of limited availability of vaccines), promoting voluntary vaccine-uptake might perform better than homogeneous vaccination. In this case, it is imperative to reduce people's fear toward vaccine uptake, especially for high-degree individuals. Properly designed educational programs could be very useful for this purpose.

Based on the modeling framework and findings of this research, there are several possible future research directions regarding vaccination game. First, this work considers population heterogeneity regarding only their contact degrees. In reality, there is a much broader spectrum of heterogeneities that should be properly addressed, such as different social groups and population spatial distribution. Moreover, in this research we consider the disease spreads quickly among a static population. When the disease spreading speed is insignificant compared to the population birth/death rates, a different dynamics model is required. Furthermore, focusing on investigating population vaccination behavior under vaccine-phobia, this study assumes the vaccine resources are always abundant. It would also be very interesting to investigate a vaccination game where people need to compete for limited supply of vaccines.

\section*{References}
\bibliography{vaccination_game_ref}

\end{document}